\RequirePackage{fix-cm}
\documentclass{svjour3}                     
\smartqed  
\usepackage{graphicx}
\usepackage{pslatex}
\usepackage{amsfonts,color}
\usepackage{amssymb,amsmath,latexsym}
\usepackage[justification=centering]{caption}

\makeatletter
\def\inmod#1{\allowbreak\mkern5mu({\operator@font mod}\,\,#1)}

\newcommand{\Rmnum}[1]{\expandafter\@slowromancap\romannumeral #1@}
\makeatother

\newtheorem{open}{Open Problem}

\begin{document}

\title{
Steiner systems $S(2, 4, \frac{3^m-1}{2})$ and $2$-designs from ternary linear   codes of length $\frac{3^m-1}{2}$
}


\author{Chunming Tang \and Cunsheng Ding \and  Maosheng Xiong }


\institute{C. Tang \at
                School of Mathematics and Information, China West Normal University, Nanchong, Sichuan,  637002, China.
              \email{tangchunmingmath@163.com} \and
C. Ding \at
                Department of Computer Science and Engineering, The
Hong Kong University of Science and Technology,
Clear Water Bay, Kowloon, Hong Kong, China.
              \email{cding@ust.hk} \and
M. Xiong \at
                Department of Mathematics, The
Hong Kong University of Science and Technology,
Clear Water Bay, Kowloon, Hong Kong, China.
              \email{mamsxiong@ust.hk}
}

\date{Received: date / Accepted: date}

\maketitle

\begin{abstract}
Coding theory and $t$-designs
have close connections and interesting interplay. In this paper,
we first introduce a class of ternary linear codes and study their parameters.
We then focus on   their three-weight subcodes with a special weight distribution.
We determine the weight distributions of some shortened codes and  punctured codes of these three-weight subcodes. These shortened and punctured codes contain some codes that have the same parameters as the best ternary linear codes known in the database maintained by Markus Grassl at http://www.codetables.de/.
These three-weight subcodes with a special weight distribution do not satisfy the conditions of the Assmus-Mattson theorem and
do not admit $2$-transitive or $2$-homogeneous automorphism groups in general.
By employing
the theory of projective geometries and projective generalized Reed-Muller codes, we prove that they still hold $2$-designs. We also determine the parameters of these $2$-designs.
This paper mainly confirms some recent conjectures of Ding and Li
regarding Steiner systems and $2$-designs from a special type of ternary projective codes.
\keywords{Cyclic code \and linear code \and  $t$-design \and Steiner system}
\subclass{94B05\and 94B15 \and 05B05}
\end{abstract}

\section{Introduction}
We start with a brief recall of $t$-designs. Let $\mathcal P$ be a set of $\nu$ elements  and $\mathcal B$   a collection of $b$ $k$-subsets of $\mathcal P$, where $\nu \ge 1$, $b\ge 1$ and  $1\le k \le \nu$.  Let $t$ be a positive integer with $t \le k$.
The pair $\mathbb  D=(\mathcal P, \mathcal B)$ is called a  $t$-$(\nu,k, \lambda )$ \emph{design}, or simply \emph{$t$-design},
if every $t$-subset of $\mathcal P$ is contained in exactly $\lambda$ elements of $\mathcal B$. The elements of $\mathcal P$ are called \emph{points},
and those of $\mathcal B$ are referred to as \emph{blocks}. It is possible for a design to have repeated blocks.
A $t$-design is called \emph{simple} if $\mathcal B$ does not contain repeated blocks.
In this paper, we consider only simple $t$-designs with $\nu > k > t$. A $t$-$(\nu, k, \lambda)$ design is called a \emph{Steiner system} denoted by
$S(t,k,\nu)$ if $t \ge 2$ and $\lambda=1$.

Linear codes and $t$-designs are companions.
A $t$-design $\mathbb  D=(\mathcal P, \mathcal B)$ induces a linear code over GF($p$) for
any prime $p$.
Let $\mathcal P=\{p_1, \dots, p_{\nu}\}$. For any block $B\in \mathcal B$,   the \emph{characteristic vector} of $B$ is defined
by the vector $\mathbf{c}_{B}
=(c_1, \dots, c_{\nu})\in \{0,1\}^{\nu}$, where
\begin{align*}
c_i=
\left\{
  \begin{array}{ll}
    1, & \text{if}~ p_i \in B, \\
    0, & \text{if}~ p_i \not \in B.
  \end{array}
\right.
\end{align*}
For a prime $p$, a \emph{linear code} $\mathsf{C}_{p}(\mathbb D)$ over the prime field $\mathrm{GF}(p)$ from the design $\mathbb D$ is spanned by the characteristic vectors of the blocks of $\mathbb B$, which is
 the subspace $\mathrm{Span}\{\mathbf{v}_{B}: B\in \mathcal B\}$ of the vector space $\mathrm{GF}(p)^{\nu}$.
Linear codes $\mathsf{C}_{p}(\mathbb D)$ from designs $\mathbb D$ have been studied and documented in the literature
 (see, for examples,
\cite{AK92,Ding15,Ton98,Ton07}).

On the other hand,  a linear code $\mathcal C$ may induce a $t$-design under certain conditions, which is formed by
the supports of  codewords of a fixed Hamming weight in $\mathcal C$.
 Let
$\mathcal P(\mathcal C)=\{0,1, \dots, n-1\}$ be the set of the coordinate positions of $\mathcal C$, where $n$ is the length of $\mathcal C$.
For a codeword $\mathbf c =(c_0, \dots, c_{n-1})$ in $\mathcal C$, the \emph{support} of  $\mathbf c$
is defined by
\begin{align*}
\mathrm{Supp}(\mathbf c) = \{i: c_i \neq 0, i \in \mathcal P(\mathcal C)\}.
\end{align*}
Let $\mathcal B_{w}(\mathcal C)
=\{\mathrm{Supp}(\mathbf c): wt(\mathbf{c})=w
~\text{and}~\mathbf{c}\in \mathcal{C}\}$. For some special $\mathcal C$, $\left (\mathcal P(\mathcal C),  \mathcal B_{w}(\mathcal C) \right)$
is a $t$-design.
In this way, many $t$-designs are derived  from linear codes
\cite{AK92,Ding18dcc,Ding18jcd,DLX17,HKM04,HMT05,KM00,MT04,Ton98,Ton07}.
A major approach to constructing $t$-designs from codes is the use of the
Assmus-Mattson Theorem \cite{AM74,HP10}.
Another major approach to constructing $t$-designs from linear codes is the use of linear codes with $t$-homogeneous or $t$-transitive automorphism groups \cite[Theorem 4.18]{Dingbk18}. Interplay between codes and designs could
be found in \cite{AK92,AK98,AM74,Ding15,Ding18dcc,Ding18jcd,Dingbk18,DLX17,HP10,KM00,MS77,MT04,Ton98,Ton07}.

In \cite{DL17}, Ding and Li constructed
some $t$-designs from linear codes and
presented some conjectured infinite families of $2$-designs from linear projective ternary  codes.
Let $m\ge 3$ be an odd integer and let $\alpha$ be  a  primitive element of $\mathrm{GF}(3^m)$.
Let $\mathcal C_1$ and $\mathcal C_2$ be linear codes defined by
\begin{align}\label{eq:C1}
\mathcal{C}_1=\left \{ \left (\mathrm{Tr}_{3^m/3}\left (a \alpha^{4i}+b \alpha^{2i} \right ) \right )_{i=0}^{\frac{3^m-1}{2}-1} :a,b\in \mathrm{GF}(3^m)\right  \},
\end{align}
and
\begin{align}\label{eq:C2}
\mathcal{C}_2=\left \{ \left (\mathrm{Tr}_{3^m/3}\left (a \alpha^{\left (3^{\frac{m-3}{2}}+1 \right )i}+b \alpha^{\left (3^{\frac{m-1}{2}}+1 \right )i} \right ) \right )_{i=0}^{\frac{3^m-1}{2}-1} :a,b\in \mathrm{GF}(3^m)\right  \},
\end{align}
where $\mathrm{Tr}_{3^m/3}(\cdot)$ is the trace function from $\mathrm{GF}(3^m)$ to $\mathrm{GF}(3)$.
Then the codes $\mathcal C_i$ ($i=1,2$) have parameters $[\frac{3^m-1}{2}, 2m, 3^{m-1}-3^{\frac{m-1}{2}}]$,
and their weight distributions are given in Table \ref{table:3W} \cite{LDXG17,Dingbk18}.
Moreover, the dual codes $\mathcal C_i^{\perp}$ of $\mathcal C_i$ ($i=1,2$)
have parameters $[\frac{3^m-1}{2}, \frac{3^m-1}{2}-2m, 4]$.

\begin{table}[htbp]
\centering
\caption{The weight distribution of $\mathcal C_1$ and $\mathcal C_2$}
\label{table:3W}
\begin{tabular}{|c|c|}
  \hline
  Weight & Frequency \\
  \hline
  $0$& $1$\\
\hline
$3^{m-1}-3^{\frac{m-1}{2}}$ & $\frac{1}{2}\cdot \left( 3^{m-1}+3^{\frac{m-1}{2}}  \right ) \left ( 3^m-1 \right )$\\
\hline
$3^{m-1}$&
 $(2\cdot 3^{m-1}+1)(3^m-1)$\\
 \hline
 $3^{m-1}+3^{\frac{m-1}{2}}$ & $\frac{1}{2}\cdot \left( 3^{m-1}-3^{\frac{m-1}{2}}  \right ) \left ( 3^m-1 \right )$\\
\hline
\end{tabular}
\end{table}

Let $\left ( A_k(\mathcal C_i) \right )_{k=0}^n$ and $\left ( A_k(\mathcal C_i^{\perp}) \right )_{k=0}^n$
denote the weight distributions of $\mathcal C_i$ and $\mathcal C_i^{\perp}$, respectively, where $i=1,2$.  In \cite{DL17}, Ding and Li made the following interesting conjectures of $2$-designs from the two linear codes.

\begin{conjecture}\label{conj:3 design}
Let $\mathcal C_i$ be the ternary code in (\ref{eq:C1}) or (\ref{eq:C2}). Let $k$ be an integer satisfying $A_k(\mathcal C_i)> 1$.
Then $\left (\mathcal P(\mathcal C_i),  \mathcal B_{k}(\mathcal C_i) \right )$ is a $2$-design.
\end{conjecture}

\begin{conjecture}\label{conj:Steiner}
Let $\mathcal C_i$ be the ternary code in (\ref{eq:C1}) or (\ref{eq:C2}).
Then $\left (\mathcal P(\mathcal C_i^{\perp}),  \mathcal B_{4}(\mathcal C_i^{\perp}) \right )$ is a
Steiner system $S(2,4,\frac{3^m-1}{2})$.
\end{conjecture}

\begin{conjecture}\label{conj:many designs}
Let $\mathcal C_i$ be the ternary code in (\ref{eq:C1}) or (\ref{eq:C2}). Let $k$ be an integer satisfying $A_k(\mathcal C_i^{\perp})> 1$.
Then $\left (\mathcal P(\mathcal C_i^{\perp}),  \mathcal B_{k}(\mathcal C_i^{\perp}) \right )$ is a $2$-design.
\end{conjecture}

The minimum distance of the dual code $\mathcal C_i^{\perp}$ is only $4$. Hence, the famous Assmus-Mattson theorem guarantees only $1$-designs supported by $\mathcal C_i$.
When $m=5$ and $7$, Magma experimental results showed that $\mathcal C_i$ is not  $2$-transitive or $2$-homogeneous.
Thus, in general, $\mathcal C_i$ is not $2$-transitive or $2$-homogeneous. Consequently,
the degree of transitivity or homogeneity of their automorphism groups cannot be employed
to proved that the two codes support $2$-designs.

The objective of this paper is to develop a different method to settle these conjectures.  Specifically, using  projective generalized Reed-Muller codes,
we prove that Conjectures \ref{conj:3 design} and \ref{conj:Steiner} hold, and Conjecture \ref{conj:many designs} is true for $k\in \{4,5,6,7\}$.

\section{Preliminaries}

In this section, we briefly recall some results  on linear codes, quadratic forms, projective geometries and projective generalized Reed-Muller codes.

\subsection{The Pless power moments}

Let $\mathcal C$ be an $[n,k,d]$ \emph{linear code} over $\mathrm{GF}(q)$. Denote by
$(A_0, A_1, \dots, A_n)$ and $(A_0^{\perp}, A_1^{\perp}, \dots, A_n^{\perp})$
 the weight distributions of $\mathcal C$ and its dual $\mathcal C^{\perp}$,
 respectively.
 The \emph{Pless power moments} \cite{HP10} are given by
 \begin{align}\label{eq:PPM}
 \sum_{i=0}^n i^t A_i= \sum_{i=0}^t (-1)^i A_i^{\perp}   \left [  \sum_{j=i}^t   j ! S(t,j) q^{k-j} (q-1)^{j-i} \binom{n-i}{n-j}  \right ],
 \end{align}
where $0\le t \le n$ and $S(t,j)=\frac{1}{j!} \sum_{i=0}^j  (-1)^{j-i} \binom{j}{i} i^t$.
These power moments can be employed to prove the following theorem \cite[Theorem 7.3.1]{HP10}.
\begin{theorem}\label{eq:wt-wtd}
Let $S\subseteq \{1,2,\dots, n\}$ with $\# S =s$. Then the weight distributions of $\mathcal C$ and $\mathcal C^{\perp}$
are uniquely determined by $A_1^{\perp}, \dots, A_{s-1}^{\perp}$ and the $A_i$ with $i\not \in S$.
These values can be found from the first $s$ equations in (\ref{eq:PPM}).
\end{theorem}

\subsection{Quadratic forms}

An $m$-ary \emph{quadratic form} over a field $\mathrm{GF}(q)$ is a homogeneous polynomial of degree $2$
in $m$ variables with coefficients in $\mathrm{GF}(q)$:
\begin{align*}
F(x_0, \dots, x_{m-1}) =\sum_{i,j=0}^{m-1} a_{i,j} x_i x_j, ~~~~~a_{i,j} \in \mathrm{GF}(q).
\end{align*}
Two quadratic forms $F$ and $F'$ over $\mathrm{GF}(q)$ are called equivalent if there is an $m$-by-$m$ nonsingular matrix $A$
such that $F'(x_0, \dots,x_{m-1})= F((x_0, \dots,x_{m-1})\cdot A)$.
If $q$ is odd, then every $m$-ary quadratic form $F(x_0, \dots, x_{m-1})$
over $\mathrm{GF}(q)$ is equivalent to the following diagonal form
$$ a_0 x_0^2+\cdots +a_{s-2} x_{s-2}^2 + a_{s-1}x_{s-1}^2,$$
where $a_i\in \mathrm{GF}(q)^{*}$ \cite[Theorem 6.21]{LN97}.
The number $s$ is called the \emph{rank} of  $F$.

For an $m$-ary quadratic form $F$ over $\mathrm{GF}(q)$, set
$$N(F=0)=\# \left \{ (x_0, \dots, x_{m-1})\in \mathrm{GF}(q)^m: F(x_0, \dots, x_{m-1}) =0 \right \}.$$
Suppose that $F$ is equivalent to the diagonal  form $a_0 x_0^2+\cdots +a_{s-2} x_{s-2}^2 + a_{s-1}x_{s-1}^2$ with  $a_i\in \mathrm{GF}(q)^{*}$.
Then, one has the following result on $N(F=0)$  \cite{LN97}:
\begin{align}\label{eq:N(F=0)}
N(F=0)=
\left\{
  \begin{array}{ll}
    q^{m-1}, & \text{ if } s \text{ is odd},\\
    q^{m-1}  + \eta(a_0 \dots a_{s-1}) \eta(-1)^{\frac{s}{2}} (q-1) q^{m-\frac{s+2}{2}}, & \text{ if } s \text{ is even,}
  \end{array}
\right.
\end{align}
 where $\eta$ denotes the quadratic character of $\mathrm{GF}(q)^{*}$.

\subsection{Designs from projective geometries}

Let $q$ be  a prime power and $m\ge 2$ be an integer.
We denote by $\mathrm{PG}(m-1,q)$ the \emph{projective geometry} of projective dimension $m-1$ over the field $\mathrm{GF}(q)$.
Its elements are the subspaces of $\mathrm{GF}(q)^{m}$ and its incidence relation is the set-theoretic inclusion.
A point of the projective geometry is given in homogeneous coordinates by $(x_0, x_1,\ldots, x_{m-1})$ where all $x_i$ are in $\mathrm{GF}(q)$ and are not all zero;
each point has $q-1$ coordinate representations, since $(ax_0,ax_1,...,ax_{m-1})$ and
$(x_0,x_1,...,x_{m-1})$
yield the same $1$-dimensional subspace of $\mathrm{GF}(q)^{m}$ for any non-zero $a\in \mathrm{GF}(q)$. The $1$-dimensional subspaces of $\mathrm{GF}(q)^{m}$ are the points of the projective geometry.
The \emph{projective dimension} of a subspace is defined to be $1$ less than the dimension of the subspace (as a vector space over $\mathrm{GF}(q)$).

For any integer $2 \le r\le m-1$, let $\mathcal P$ be the set of points of $\mathrm{PG}(m-1,q)$
and  $\mathcal B$  the collection  of  all the subspaces  of projective dimension $r-1$ of  $\mathrm{PG}(m-1,q)$.
Since the action of the projective general linear group on the points of $\mathrm{PG}(m-1,q)$  is  doubly-transitive,
$(\mathcal P, \mathcal B)$ is a $2$-design with parameters depending on $m$, $r$ and $q$, which will be denoted by $\mathbb{PD}(r-1,m-1,q)$.
The $2$-design $\mathbb{PD}(1,m-1,q)$ is  a Steiner system $S(2,q+1,\frac{q^m-1}{q-1})$.

\subsection{Projective generalized Reed-Muller codes}\label{subsec:RM}

For an  inter $r \ge 0$, let $\mathrm{PP}(r,m-1,q)$ be the linear subspace of $\mathrm{GF}(q)[x_0,x_1, \dots, x_{m-1}]$,
which is spanned by all
monomial $x_0^{i_0}x_1^{i_1}\cdots x_{m-1}^{i_{m-1}}$ satisfying the following two conditions
\begin{itemize}
\item $\sum_{j=0}^{m-1} i_j \equiv 0 \pmod{q-1}$,
\item    $0<\sum_{j=0}^{m-1} i_j   \le r(q-1)$.
\end{itemize}
Each $a \in \mathrm{GF}(q)$ is viewed as the constant function $f_a(x_0, x_1, \ldots, x_{m-1}) \equiv a$.

Let $\{\mathbf{x}^1, \dots, \mathbf{x}^{N}\}$
be the set of projective points in $\mathrm{PG}(m-1,q)$, where $N=\frac{q^m-1}{q-1}$. Then, the \emph{$r^{th}$ order
 projective generalized Reed-Muller code} $\mathrm{PRM}(r,m-1,q)$ of length $\frac{q^m-1}{q-1}$ is given as
\begin{align*}
\mathrm{PRM}(r,m-1,q)=\left \{\left (f(\mathbf{x}^1), \dots, f(\mathbf{x}^N) \right ): f\in \mathrm{PP}(r,m-1,q)\cup \mathrm{GF}(q) \right \}.
\end{align*}
When $r>1$, let  $\mathrm{PRM}^*(r,m-1,q)$ be the subcode of $\mathrm{PRM}(r,m-1,q)$
defined by
\begin{align*}
\mathrm{PRM}^*(r,m-1,q)=\left \{\left (f(\mathbf{x}^1), \dots, f(\mathbf{x}^N) \right ): f\in \mathrm{PP}(r,m-1,q) \right \}.
\end{align*}
Thus, $\mathrm{PRM}^*(r,m-1,q)$ is  the even-like subcode of $\mathrm{PRM}(r,m-1,q)$.
For the minimum weight and the dual of the projective generalized Reed-Muller code, we have the following \cite{AK98}.
\begin{theorem}
Let $0 \le r \le m-1$. Then, the minimal weight of $\mathrm{PRM}(r,m-1,q)$ is $\frac{q^{m-r}-1}{q-1}$ and
\begin{align*}
\mathrm{PRM}(r,m-1,q)^{\perp} = \mathrm{PRM}^*(m-1-r,m-1,q).
\end{align*}
\end{theorem}

The following theorem gives the relationship between the codes
$\mathsf{C}_{p}\left (\mathbb{PD}(r-1,m-1,q) \right )$ of designs
arising from projective geometries and the  projective generalized Reed-Muller codes \cite{AK98}.

\begin{theorem}\label{thm:PG-PRM}
Let $m$ be a positive integer, $p$ a prime, and $0 \le r \le m$.

\rm(i) The code $\mathsf{C}_{p}\left (\mathbb{PD}(r-1,m-1,p) \right )$
from the design of points and projective
$(r-1)$-dimensional subspaces of the projective geometry $\mathrm{PG}(m-1,p)$ is the same as $\mathrm{PRM}(m-r,m-1,p)$
 up to a permutation of coordinates.

 \rm(ii) $\mathsf{C}_{p}\left (\mathbb{PD}(r-1,m-1,p) \right )$ has minimum weight
 $\frac{p^r-1}{p-1}$ and the minimum-weight vectors are the multiples of the characteristic  vectors of the blocks.

 \rm(iii) The dual code $\mathsf{C}_{p}\left (\mathbb{PD}(r-1,m-1,p) \right )^{\perp}$
 of $\mathsf{C}_{p}\left (\mathbb{PD}(r-1,m-1,p) \right )$ is the same as $\mathrm{PRM}^*(r-1,m-1,p)$
 up to a permutation of coordinates and has minimum weight at least  $\frac{p^{m-r+1}-1}{p-1}+1$.

\rm(iv) The dimension of the code $\mathsf{C}_{p}\left (\mathbb{PD}(r-1,m-1,p) \right )$ is
$$ \frac{p^m-1}{p-1}-\sum_{i=0}^{r-2} (-1)^{i} \binom{(r-1-i)(p-1)-1}{i} \binom{m-r+(r-1-i)p}{m-1-i}.$$
\end{theorem}

\section{A class of ternary cyclic codes of length $\frac{3^m-1}{2}$}

In this section, we   introduce a  class of ternary cyclic codes of length $\frac{3^m-1}{2}$, and determine the parameters of these codes and their dual codes.

Let $m$ be a positive integer, $\alpha$ be a generator of $\mathrm{GF}(3^m)^*$ and $\beta=\alpha^{2}$. Then $\beta$ is a primitive $\frac{3^m-1}{2}$-th root of unity
in $\mathrm{GF}(3^m)$. Let $l$ be a positive integer and  $E=\{k_i : 0 \le i \le l \}$  a set of integers with
$0\le k_0<k_1<\cdots<k_l\le \frac{m}{2}$.
A \emph{ternary cyclic code $\mathcal C(E)$} of  length $n=\frac{3^m-1}{2}$ is defined by
\begin{equation}\label{eq:C-E}
\mathcal{C}(E)=\{\mathbf{c}(\underline{a}):
\underline{a}=(a_0, \dots, a_l) \in \mathrm{GF}(3^m)^{l+1}\},
\end{equation}
where $\mathbf{c}(\underline{a})=\left( \mathrm{Tr}_{3^m/3}\left (\sum_{j=0}^l a_j \alpha^{\left (3^{k_j}+1 \right )i} \right ) \right )_{i=0}^{n-1}$.
We also write $\mathcal C(k_0,k_1,\dots, k_l)$ for  $\mathcal C(E)$.
From Delsarte's theorem \cite{MS77}, the dual code $\mathcal C(E)^{\perp}$ of $\mathcal C(E)$ can be given by
\begin{align*}
\mathcal C(E)^{\perp}=\left \{(w_0, \dots, w_{n-1}) \in \mathrm{GF}(3)^{n}: \sum_{i=0}^{n-1} w_i \mathbf u_i=\mathrm{0}\right \},
\end{align*}
where $\mathbf u_i=\left( \alpha^{i\left (3^{k_j}+1 \right )} \right )_{j=0}^{l} \in \mathrm{GF}(3^m)^{l+1}$.
Note that for the linear codes  defined in (\ref{eq:C1}) and (\ref{eq:C2}), $\mathcal C_1=\mathcal C(0,1)$ and $\mathcal C_2=\mathcal{C} (\frac{m-3}{2},\frac{m-1}{2})$.

To determine the parameters of $\mathcal C(E)$ and its dual, we need some results on irreducible polynomials.
For an  integer $e$, let $\mathbb M_{\alpha^{-e}}(x)\in \mathrm{GF}(3)[x]$ be the minimal polynomial of $\alpha^{-e}$ over $\mathrm{GF}(3)$. We have the following lemma on
$\mathbb M_{\alpha^{-e}}(x)$.
\begin{lemma}\label{lem:irre-poly}
Let $k$ and $k'$ be two integers such that $0\le k, k' \le \frac{m}{2}$. Then

\rm{(i)} $\mathrm{deg}\left (\mathbb M_{\alpha^{-\left (3^{k}+1 \right)}}(x)\right )=
\left\{
  \begin{array}{ll}
    m, & \text{if}~k<\frac{m}{2}     ,  \\
    \frac{m}{2}, &  \text{if}~\text{ m is even and } k=\frac{m}{2}.
  \end{array}
\right.
$

\rm{(ii)} $\mathbb M_{\alpha^{-\left (3^{k}+1 \right)}}(x)=\mathbb M_{\alpha^{-\left (3^{k'}+1 \right)}}(x)$
if and only if $k=k'$.
\end{lemma}
\begin{proof}
\rm(i) Note that  $\mathrm{deg}\left (\mathbb M_{\alpha^{-\left (3^{k}+1 \right)}}(x)\right )$ is the least positive integer $d\le m$ such that
\begin{align}\label{eq:cyc-3^m-1}
(3^k+1)3^d\equiv (3^k+1) \pmod{3^m-1}.
\end{align}
It suffices to prove that Equation (\ref{eq:cyc-3^m-1}) holds  if and only if $d=k=\frac{m}{2}$ or $d=m$.
If $d=k=\frac{m}{2}$ or $d=m$, we can easily verify   that Equation (\ref{eq:cyc-3^m-1}) holds.
Conversely, suppose that $(3^k+1)3^d\equiv (3^k+1) \pmod{3^m-1}$ and $d<m$.
If $k+d<m$, then  $3^{k+d}+3^d=3^k+1$ and $d=0$,    which leads to a contradiction with $d>0$.
If $k+d \ge m$, then  $(3^k+1)3^d\equiv 3^{k+d-m}+3^d \pmod{3^m-1}$.
From $0\le k+d-m \le \frac{m}{2}$, we have
$3^{k+d-m}+3^d=3^k+1$, i.e.,  $d=k$, $k+d-m=0$,
and $d=k=\frac{m}{2}$. Part \rm (i) follows.

\rm(ii) Suppose $\mathbb M_{\alpha^{-\left (3^{k}+1 \right)}}(x)=\mathbb M_{\alpha^{-\left (3^{k'}+1 \right)}}(x)$.
This holds if and only if there exists a positive integer $d\le m$ such that $(3^k+1)3^d\equiv 3^{k'}+1 \pmod{3^m-1}$.
If $d=m$, then $k=k'$. If $d<m$ and $k+d<m$,
then we have $3^{k+d}+3^d=3^{k'}+1$, which  contradicts the fact that $d\ge 1$.
If $d<m$ and $k+d\ge m$, then we get $3^{k+d-m}+3^d \equiv 3^{k'}+1 \pmod{3^m-1}$.
Hence $k'=d$ and $k=m-d$. From $k,k'\le \frac{m}{2}$, one obtains $k'=k=d=\frac{m}{2}$.
Part \rm(ii) follows.
\end{proof}

As a consequence of Lemma \ref{lem:irre-poly}, we have the following proposition on
some parameters of $\mathcal{C}(E)$.

\begin{proposition}\label{prop:cyc-dim}
The linear code $\mathcal{C}(E)$ defined by (\ref{eq:C-E}) is a ternary cyclic
code of length $\frac{3^m-1}{2}$ and dimension
$$
\mathrm{dim} (\mathcal{C}(E))=
\left\{
  \begin{array}{ll}
    (l+1)m, & \text{if} ~k_l<\frac{m}{2}, \\
    \frac{2l+1}{2}m, &\text{if} ~ m \text{ is even and } k_l=\frac{m}{2}.
  \end{array}
\right.
$$
In particular,  the linear code $\mathcal{C}\left (0,1,2, \dots, \left \lfloor \frac{m}{2} \right \rfloor     \right )$
is a cyclic code with dimension $\frac{m(m+1)}{2}$, where  $\left \lfloor \frac{m}{2} \right \rfloor$ is
the greatest integer less than or equal to $\frac{m}{2}$.
\end{proposition}

In the following, we will  determine the parameters of the codes
$\mathcal{C}\left (0,1,2, \dots, \left \lfloor \frac{m}{2} \right \rfloor     \right )$ and
$\mathcal{C}\left (0,1, \dots, \lfloor \frac{m}{2}\rfloor \right  )^{\perp}$.

Let $\{\alpha_0, \dots, \alpha_{m-1}\}$ be a basis of $\mathrm{GF}(3^m)$ over $\mathrm{GF}(3)$.
Let $\rho$ be the linear transformation
from $\mathrm{GF}(3^m)$ to $\mathrm{GF}(3)^m$ defined by
\begin{align*}
\rho (x) = (x_0,x_1, \dots, x_{m-1})\in \mathrm{GF}(3)^{m},
\end{align*}
 where $x=\sum_{i=0}^{m-1} x_i \alpha_i
\in \mathrm{GF}(3^{m})$.
From this isomorphism $\rho$, a function $f: \mathrm{GF}(3^m) \rightarrow \mathrm{GF}(3)$
induces a function $F: \mathrm{GF}(3)^m \rightarrow \mathrm{GF}(3)$. In particular, the
function
$$f_{a_0, \dots, a_{\lfloor \frac{m}{2} \rfloor}}(x)=\mathrm{Tr}_{3^m/3} \left( \sum_{i=0}^{\lfloor \frac{m}{2} \rfloor  } a_i x^{3^i+1} \right )$$
with $a_i \in \mathrm{GF}(3^m)$ induces a quadratic form
\begin{align}\label{eq:f-F}
&F_{a_0, \dots, a_{\lfloor \frac{m}{2} \rfloor}}(x_0, \dots, x_{m-1})\nonumber \\
&=  \mathrm{Tr}_{3^m/3} \left ( \sum_{t=0}^{\lfloor \frac{m}{2} \rfloor} a_t \left (\sum_{i=0}^{m-1} x_i \alpha_i \right )^{3^t+1} \right )\nonumber\\
&= \mathrm{Tr}_{3^m/3} \left ( \sum_{t=0}^{\lfloor \frac{m}{2} \rfloor} a_t \left (\sum_{i=0}^{m-1} x_i \alpha_i^{3^t} \right ) \left (\sum_{j=0}^{m-1} x_j \alpha_j \right ) \right )\nonumber
\\
&= \sum_{i,j=0}^{m-1}  \mathrm{Tr}_{3^m/3} \left ( \sum_{t=0}^{\lfloor \frac{m}{2} \rfloor}  a_t \alpha_i^{3^t} \alpha_j \right )    x_ix_j \in \mathrm{PP}(1,m-1,3),
\end{align}
where $\mathrm{PP}(1,m-1,3)$ was defined in Subsection \ref{subsec:RM}. From the definition of $F_{a_0, \dots, a_{\lfloor \frac{m}{2} \rfloor}}$,
we have $F_{a_0, \dots, a_{\lfloor \frac{m}{2} \rfloor}}(\rho(x))= f_{a_0, \dots, a_{\lfloor \frac{m}{2} \rfloor}}(x)$ for any $ x \in \mathrm{GF}(3^m)$.

Note that the set of all projective points of $\mathrm{PG}(m-1,3)$ is $\left \{\rho \left (\alpha^{i} \right ): i=0, \dots, \frac{3^m-1}{2}-1 \right \}$, where
$\alpha$ is a generator of $\mathrm{GF}(3^m)^*$.
Hence, we can choose $\mathbf x^i= \rho \left (\alpha^{i} \right )$ in the definition of $\mathrm{PRM}(1,m-1,3)$ and $\mathrm{PRM}^*(1,m-1,3)$.
A map $\pi$ from $\mathcal{C}\left (0,1, \dots, \lfloor \frac{m}{2}\rfloor \right  )$ to
$\mathrm{PRM}^*(1,m-1,3)$ can be defined by
\begin{align*}
\pi: \mathcal{C}\left (0,1, \dots, \lfloor \frac{m}{2}\rfloor \right  ) & \longrightarrow \mathrm{PRM}^*(1,m-1,3)\\
\left (f_{a_0, \dots, a_{\lfloor \frac{m}{2} \rfloor}}\left (\alpha^i \right ) \right )_{i=0}^{\frac{3^m-1}{2}-1} & \longmapsto   \left (F_{a_0, \dots, a_{\lfloor \frac{m}{2} \rfloor}}\left (\rho(\alpha^i) \right ) \right )_{i=0}^{\frac{3^m-1}{2}-1},
\end{align*}
where $f_{a_0, \dots, a_{\lfloor \frac{m}{2} \rfloor}}(x)=\mathrm{Tr}_{3^m/3} \left( \sum_{i=0}^{\lfloor \frac{m}{2} \rfloor  } a_i x^{3^i+1} \right )$
with $a_i \in \mathrm{GF}(3^m)$  and  $F_{a_0, \dots, a_{\lfloor \frac{m}{2} \rfloor}}$ was  defined in (\ref{eq:f-F}).
Since $F_{a_0, \dots, a_{\lfloor \frac{m}{2} \rfloor}}(\rho(\alpha^i))= f_{a_0, \dots, a_{\lfloor \frac{m}{2} \rfloor}}(\alpha^i)$ and  $\pi$ is
an inclusion map, we have
$$
\mathcal{C}\left (0,1, \dots, \lfloor \frac{m}{2}\rfloor \right  )  \subseteq \mathrm{PRM}^*(1,m-1,3).
$$
On the other hand, from Theorem \ref{thm:PG-PRM} and Proposition \ref{prop:cyc-dim}, we have the demension
$$
\mathrm{dim}\left (\mathcal{C}\left (0,1, \dots, \lfloor \frac{m}{2}\rfloor \right  ) \right )=
\mathrm{dim}\left (\mathrm{PRM}^*(1,m-1,3) \right )=\frac{m(m+1)}{2}.$$
Using Theorem \ref{thm:PG-PRM} again, we have
\begin{align*}
\mathcal{C}\left (0,1, \dots, \lfloor \frac{m}{2}\rfloor \right  )  =   \mathsf{C}_{3}\left (\mathbb{PD}(1,m-1,3) \right )^{\perp}=\mathrm{PRM}^*(1,m-1,3).
\end{align*}
Note that $\mathbb{PD}(1,m-1,3)$ is a Steiner system $S(2,4, \frac{3^m-1}{2})$ with $\frac{(3^m-1)(3^{m-1}-1)}{16}$ blocks.
From the previous discussion and Theorem \ref{thm:PG-PRM}, we have the following
theorem.
\begin{theorem}\label{thm:C(0,1,...,m/2)}
Let $m\ge 3$ be an integer.

\rm(i) $\mathcal{C}\left (0,1, \dots, \lfloor \frac{m}{2}\rfloor \right  ) =\mathrm{PRM}^*(1,m-1,3)$, that is, $\mathcal{C}\left (0,1, \dots, \lfloor \frac{m}{2}\rfloor \right  )$
is the even-like subcode of the first order projective generalized Reed-Muller code $\mathrm{PRM}(1,m-1,3)$.

\rm(ii) $\mathcal{C}\left (0,1, \dots, \lfloor \frac{m}{2}\rfloor \right  )^{\perp}$ is the code $\mathsf{C}_{3}\left (\mathbb{PD}(1,m-1,3) \right )$
of the Steiner system $\mathbb{PD}(1,m-1,3)$.

\rm(iii)  $\mathcal{C}\left (0,1, \dots, \lfloor \frac{m}{2}\rfloor \right  )^{\perp}$ has minimum distance $4$
and the minimum-weight codewords  are the multiples of the characteristic  vectors of the blocks of $\mathbb{PD}(1,m-1,3)$.

\rm(iv) The number $A_4^{\perp}$ of codewords with Hamming weight $4$ of $\mathcal{C}\left (0,1, \dots, \lfloor \frac{m}{2}\rfloor \right  )^{\perp}$
is $\frac{(3^m-1)(3^{m-1}-1)}{8}$.
\end{theorem}
\begin{remark}
Form the definition of $\mathbb{PD}(1,m-1,3)$, the Steiner system $\mathbb{PD}(1,m-1,3)$ is also equivalent to
the Steiner system $(\mathcal P, \mathcal B)$, where
\begin{align*}
\mathcal {P}=\{a^2: a\in \mathrm{GF}(3^m)^*\}
\end{align*}
and
\begin{align*}
\mathcal{B}=\left \{\{a^2,b^2, (a+b)^2,(a-b)^2\}: a, b\in \mathrm{GF}(3^m)^* \text{ and } a\neq \pm b \right \}.
\end{align*}
\end{remark}

The minimum distance of the code $\mathcal{C}\left (0,1, \dots, \lfloor \frac{m}{2}\rfloor \right  )$ is described in the next theorem.

\begin{theorem}
Let $m$ be an integer with $m\ge 3$. Then, $\mathcal{C}\left (0,1, \dots, \lfloor \frac{m}{2}\rfloor \right  )$ is a
$[\frac{3^m-1}{2}, \frac{m(m+1)}{2}, 2\cdot 3^{m-2}]$ cyclic code.
\end{theorem}
\begin{proof}
From Part (i) of  Theorem \ref{thm:C(0,1,...,m/2)}, for any nonzero codeword $\mathbf{c}=(c_0, \dots, c_{n-1})$, there is a unique  quadratic form
$F\in \mathrm{PP}(1,m-1,3)$ such that $c_i=F(\rho(\alpha^i))$, where $n=\frac{3^m-1}{2}$.
Then,
\begin{align*}
\mathrm{wt}(\mathbf{c})=\frac{3^m-N(F=0)}{2},
\end{align*}
where $F$ was defined in (\ref{eq:N(F=0)}) and $\mathrm{wt}(\mathbf{c})$ is the Hamming weight of $\mathbf{c}$. Suppose that $F$ is equivalent to
the diagonal  form $a_0 x_0^2+\cdots +a_{s-2} x_{s-2}^2 + a_{s-1}x_{s-1}^2$ with  $a_i\in \mathrm{GF}(3)^{*}$.
Using (\ref{eq:N(F=0)}), one gets
\begin{align*}
\mathrm{wt}(\mathbf{c})=
\left\{
  \begin{array}{ll}
   3^{m-1}, & \text{ if } s\equiv 1 \pmod{2}, \\
   3^{m-1}\pm 3^{m-1-\frac{s}{2}}, & \text{ if } s\equiv 0 \pmod{2}.
  \end{array}
\right.
\end{align*}
Since $F\neq 0$, then $s\ge 1$. Thus, $\mathrm{wt}(\mathbf{c}) \ge 3^{m-1}-3^{m-1-1}=2\cdot 3^{m-2}$.
In addition, choose $F=x_0^2-x_1^2$ and $\mathbf{c}=\left( F(\rho(\alpha^i)) \right )_{i=0}^{n-1}$.
Then, $\mathrm{wt}(\mathbf{c})=2\cdot 3^{m-2}$. Hence, the minimum weight of $\mathcal{C}\left (0,1, \dots, \lfloor \frac{m}{2}\rfloor \right  )$
is $2\cdot 3^{m-2}$. From Proposition \ref{prop:cyc-dim}, this theorem follows.
\end{proof}

\section{Shortened codes and punctured codes from $\mathcal C(E)$}

In this section, we present some ternary codes by shortening and puncturing
some subcodes of $\mathcal {C} \left (0,1,\dots, \lfloor  \frac{m}{2} \rfloor \right)$
with the weight distribution in Table  \ref{table:3W}, and determine their weight distributions.

Let $\mathcal  C$ be an $[n,k,d]$ code and $T$  a  set of $t$ coordinate positions in $\mathcal  C$.
We use $\mathcal  C^T$ to denote the code obtained by puncturing $\mathcal  C$  on
$T$, which is called the  \emph{punctured code}  of $\mathcal C$ on $T$.
Let $\mathcal C(T)$ be the subcode of $\mathcal C$, which is the set of codewords which are
$\mathbf{0}$ on $T$.
We now puncture $\mathcal C(T)$ on $T$, and obtain a linear code $\mathcal C_{T}$, which is called the \emph{shortened code} of $\mathcal C$ on $T$.
We have the following result on the punctured code and shortened code of a code $\mathcal{C}$ \cite{HP10}:
\begin{align}\label{eq:sh-puct}
 \left (\mathcal C_{T} \right )^{\perp} = \left ( \mathcal C^{\perp} \right )^{T}.
\end{align}

To determine the weight distributions of shortened codes
from $\mathcal C(E)$, we will need the following lemma.

\begin{lemma}\label{lem:C-T}
Let $\mathcal C$ be an $[n,k,d]$ code over $\mathrm{GF}(q)$ and  $d^{\perp}$  the minimum distance of $\mathcal  C^{\perp}$.
Let $i_1, \dots, i_s$ be $s$ positive integers and $T$ a set of $t$ coordinate positions of $\mathcal C$,  where $i_1<\cdots <i_s\le n$ and  $t<d^{\perp}$. Suppose that  $A_i(\mathcal  C)=0$ for any $i\not \in \{0, i_1, \dots, i_s \}$ and $A_1(\left ( \mathcal C^{\perp} \right )^{T})$, $\dots$,  $A_{s-1}(\left ( \mathcal C^{\perp} \right )^{T})$ are independent of the elements of $T$.
Then, the code $\mathcal C_{T}$ has dimension $k-t$.
Furthermore, the weight distributions of $\mathcal C_{T}$ and $\left ( \mathcal C^{\perp} \right )^{T}$ are independent of
the elements of $T$ and can be determined from the first $s$ equations in (\ref{eq:PPM}).
\end{lemma}

\begin{proof}
Form \cite{HP10}, $\mathcal C_{T}$ has dimension $k-t$, and $\left (\mathcal C_{T} \right )^{\perp} = \left ( \mathcal C^{\perp} \right )^{T}$. The desired conclusions of this lemma
then follow from Theorem \ref{eq:wt-wtd}.
\end{proof}

The next lemma will be useful in the sequel \cite[Lemma 8.4]{Dingbk18}.

\begin{lemma}\label{eq:dual-abc}
Let $m\ge 3$ be odd. Let  $\mathcal C$ be a  subcode of $\mathcal C \left (0,1, \dots, \lfloor \frac{m}{2} \rfloor \right )$
with the weight distribution in Table  \ref{table:3W}. Then, the weight distribution $A_1^{\perp}, \dots, A_{(3^m-1)/2}^{\perp}$ of $\mathcal C^{\perp}$
is given by
\begin{align*}
3^{2m}A_{k}^{\perp}= & \sum_{i=0}^{k}
 (-1)^{i}2^{k-i} a \binom{3^{m-1}-3^{(m-1)/2}}{i}  \binom{\frac{   3^{m-1} +2\cdot  3^{(m-1)/2} -1 }{2} }{k-i} \\
 & +\binom{\frac{3^m-1}{2}}{k} 2^k   +       \sum_{i=0 }^k
 (-1)^{i}2^{k-i} b \binom{3^{m-1}}{i}  \binom{\frac{   3^{m-1}  -1 }{2} }{k-i} \\
& +   \sum_{i=0}^k
 (-1)^{i}2^{k-i} c \binom{3^{m-1}+3^{(m-1)/2}}{i}  \binom{\frac{   3^{m-1} -2\cdot  3^{(m-1)/2} -1 }{2} }{k-i}
\end{align*}
for $0\le k \le \frac{3^m-1}{2}$, where
\begin{align*}
a=&\frac{1}{2}\cdot \left( 3^{m-1}+3^{\frac{m-1}{2}}  \right ) \left ( 3^m-1 \right ),\\
b=&(3^m-3^{m-1}+1)(3^m-1),\\
c=& \frac{1}{2}\cdot \left( 3^{m-1}-3^{\frac{m-1}{2}}  \right ) \left ( 3^m-1 \right ).
\end{align*}

In addition, $\mathcal C^{\perp}$ has parameters $[\frac{3^m-1}{2}, \frac{3^m-1}{2}-2m, 4]$ and $A_4^{\perp}=\frac{(3^m-1)(3^{m-1}-1)}{8}$.
\end{lemma}

Let $W_i(\mathcal C)$ denote the set of codewords of weight $i$ in a code $\mathcal  C$.

\begin{lemma}\label{lem:A12=0}
Let $m\ge 3$ be odd and $\mathcal C$ be a  subcode of $\mathcal C \left (0,1, \dots, \lfloor \frac{m}{2} \rfloor \right )$
with the weight distribution in Table \ref{table:3W}.  Let $T$ be a  set of $t$ coordinate positions in $\mathcal  C$.

\rm{(i)} If $t=1$, then $A_1\left ( \left (\mathcal C^{\perp} \right )^{T}  \right )=A_2\left ( \left (\mathcal C^{\perp} \right )^{T}  \right )=0$.

\rm{(ii)} If $t=2$, then  $A_1\left ( \left (\mathcal C^{\perp} \right )^{T}  \right )=0$
and $A_2\left ( \left (\mathcal C^{\perp} \right )^{T}  \right )=2$.

\rm(iii) $W_4\left( \mathcal C^{\perp} \right)= W_4\left( \mathcal C \left (0,1, \dots, \lfloor \frac{m}{2} \rfloor \right )^{\perp} \right)$.
\end{lemma}

\begin{proof}
By Lemma \ref{eq:dual-abc}, the minimum weight of $\mathcal C^{\perp}$ is $4$. Thus,
$A_1\left ( \left (\mathcal C^{\perp} \right )^{T}  \right )=A_2\left ( \left (\mathcal C^{\perp} \right )^{T}  \right )=0$ for $t=1$
and $A_1\left ( \left (\mathcal C^{\perp} \right )^{T}  \right )=0$ for $t=2$.
Since $\mathcal C \subseteq \mathcal C\left (0,1, \dots, \lfloor \frac{m}{2} \rfloor \right )$,
$ \mathcal C\left (0,1, \dots, \lfloor \frac{m}{2} \rfloor \right )^{\perp} \subseteq \mathcal C^{\perp}$
and $W_4\left ( \mathcal C\left (0,1, \dots, \lfloor \frac{m}{2} \rfloor \right )^{\perp} \right ) \subseteq W_4\left (\mathcal C^{\perp} \right )$.
Combining Part (iv) of Theorem \ref{thm:C(0,1,...,m/2)} and Lemma \ref{eq:dual-abc}, one obtains
$\# W_4\left ( \mathcal C\left (0,1, \dots, \lfloor \frac{m}{2} \rfloor \right )^{\perp} \right ) =\# W_4\left (\mathcal C^{\perp} \right )$.
As a result,
$$
W_4\left ( \mathcal C\left (0,1, \dots, \lfloor \frac{m}{2} \rfloor \right )^{\perp} \right ) = W_4\left (\mathcal C^{\perp} \right ).
$$
By Part (iii) of Theorem \ref{thm:C(0,1,...,m/2)}, $ W_4\left (\mathcal C^{\perp} \right )$
is the set of  the multiples of the characteristic  vectors of the blocks of  $\mathbb{PD}(1,m-1,3)$.
Since $\mathbb{PD}(1,m-1,3)$ is a Steiner system $S(2,4, \frac{3^m-1}{2})$, $A_2\left ( \left (\mathcal C^{\perp} \right )^{T}  \right )=2$.
This completes the proof.
\end{proof}

For $T=\{t\}$ and
$T=\{t_1,t_2\}$, we determine the weight distribution of the shortened code $\mathcal C_{T}$  of some subcodes of $\mathcal C \left (0,1, \dots, \lfloor \frac{m}{2} \rfloor \right )$.

\begin{theorem}\label{thm:sc1}
Let $t$ be an integer and  $m\ge 3$ odd,
where $0 \le t \le \frac{3^m-1}{2}-1$. Let $\mathcal C$ be a subcode of $\mathcal C \left (0,1, \dots, \lfloor \frac{m}{2} \rfloor \right )$
with the weight distribution in Table  \ref{table:3W}.
Then, the shortened code $\mathcal C_{\{t\}}$ is a ternary linear code of length $\frac{3^m-1}{2}-1$ and dimension $2m-1$, and has the weight distribution in Table \ref{table:3W-1}.
\begin{table}[htbp]
\centering
\caption{The weight distribution of the shortened code  $\mathcal C_{\{t\}}$}
\label{table:3W-1}
\begin{tabular}{|c|c|}
  \hline
  Weight & Frequency \\
  \hline
  $0$& $1$\\
\hline
$3^{m-1}-3^{\frac{m-1}{2}}$ & $ \frac{1}{2} \cdot \left ( 3^{m-1} +2\cdot 3^{\frac{m-1}{2}}-1  \right ) \left ( 3^{m-1} +3^{\frac{m-1}{2}}\right )     $\\
\hline
$3^{m-1}$&
 $\left ( 2\cdot 3^{m-1} + 1\right ) \left ( 3^{m-1} -1\right )  $\\
 \hline
 $3^{m-1}+3^{\frac{m-1}{2}}$ & $ \frac{1}{2} \cdot \left ( 3^{m-1} -2\cdot 3^{\frac{m-1}{2}}-1  \right ) \left ( 3^{m-1} -3^{\frac{m-1}{2}}\right )      $\\
\hline
\end{tabular}
\end{table}

\end{theorem}

\begin{proof}
It follows from Lemma \ref{eq:dual-abc} that $d(\mathcal C^\perp)=4$. By Lemma \ref{lem:C-T}, $\mathcal C_{\{t\}}$ has length $n=\frac{3^m-1}{2}-1$ and dimension $k=2m-1$.
Note that $A_i=A_i\left (\mathcal C_{\{t\}} \right )=0$ for $i\not \in \{0, i_1, i_2, i_3\} $, where $i_1=3^{m-1}-3^{\frac{m-1}{2}}$, $i_2=3^{m-1}$
and $i_3=3^{m-1}+3^{\frac{m-1}{2}}$.
It follows from (\ref{eq:sh-puct}) and Lemma \ref{lem:A12=0} that
$$
A_1(\left (\mathcal C_{\{t\}} \right )^\perp)=A_2(\left (\mathcal C_{\{t\}} \right )^\perp)=0.
$$
The first three Pless power moments in (\ref{eq:PPM}) give
\begin{align*}
\left\{
  \begin{array}{l}
    A_{i_1} + A_{i_2} +A_{i_3} = 3^{2m-1}-1,  \\
    i_1 A_{i_1} + i_2  A_{i_2} + i_3 A_{i_3}  = 2\cdot 3^{2m-1-1}n, \\
    i_1^2 A_{i_1} + i_2^2  A_{i_2} + i_3^2 A_{i_3}  = 2\cdot 3^{2m-1-2}n(2n+1).
  \end{array}
\right.
\end{align*}
Solving this system of equations yields the weight distribution in Table \ref{table:3W-1}.
\end{proof}

\begin{example}
Let $m=5$, $E=\{0,1\}$, $0\le t \le \frac{3^m-1}{2}-1$ and $\mathcal C=\mathcal {C}(E)$. Then $\mathcal C$ has the weight distribution in Table \ref{table:3W}.
Furthermore, the shortened code $\mathcal C_{\{t\}}$ has parameters $[120,9,72]$ and weight enumerator $1+ 4410z^{72}+13040z^{81}+2232z^{90}$.
This code has the same parameters as the best ternary linear code known in the database maintained  by Markus Grassl.

Magma experiments showed that all the shortened codes $\mathcal C_{\{t\}}$ have the same
weight distribution and are pairwise equivalent.
\end{example}

\begin{theorem}\label{thm:sc2}
Let $t_1$ and $t_2$ be two integers and  $m\ge 3$   odd, where $0 \le t_1 < t_2 \le \frac{3^m-1}{2}-1$.
Let $\mathcal C$ be a subcode of $\mathcal C \left (0,1, \dots, \lfloor \frac{m}{2} \rfloor \right )$
with the weight distribution in Table  \ref{table:3W}.
Then, the shortened code $\mathcal C_{\{t_1, t_2\}}$ is a ternary linear code of length $\frac{3^m-1}{2}-2$ and dimension $2m-2$, and has the weight distribution in Table \ref{table:3W-2}.

\begin{table}[htbp]
\centering
\caption{The weight distribution of the shortened code $\mathcal C_{\{t_1, t_2\}}$}
\label{table:3W-2}
\begin{tabular}{|c|c|}
  \hline
  Weight & Frequency \\
  \hline
  $0$& $1$\\
\hline
$3^{m-1}-3^{\frac{m-1}{2}}$ & $ \frac{1}{6} \cdot \left ( 3^{m-1} +2\cdot 3^{\frac{m-1}{2}}-1  \right ) \left ( 3^{m-1} +3^{\frac{m+1}{2}}\right )     $\\
\hline
$3^{m-1}$&
 $\left ( 2\cdot 3^{m-1} + 1\right ) \left ( 3^{m-2} -1\right )  $\\
 \hline
 $3^{m-1}+3^{\frac{m-1}{2}}$ & $ \frac{1}{6} \cdot \left ( 3^{m-1} -2\cdot 3^{\frac{m-1}{2}}-1  \right ) \left ( 3^{m-1} -3^{\frac{m+1}{2}}\right )      $\\
\hline
\end{tabular}
\end{table}

\end{theorem}

\begin{proof}
By Lemma \ref{eq:dual-abc}, $d(\mathcal C^\perp)=4$. It follows from Lemma \ref{lem:C-T} that
$\mathcal C_{\{t_1,t_2\}}$ has length $n=\frac{3^m-1}{2}-2$ and dimension $k=2m-2$.
Note that $A_i=A_i\left (\mathcal C_{\{t_1,t_2\}} \right )=0$ for $i\not \in \{0, i_1, i_2, i_3\} $, where $i_1=3^{m-1}-3^{\frac{m-1}{2}}$, $i_2=3^{m-1}$
and $i_3=3^{m-1}+3^{\frac{m-1}{2}}$.
Combining (\ref{eq:sh-puct}) and Lemma \ref{lem:A12=0},
we deduce that
$A_1(\left (\mathcal C_{\{t_1,t_2\}} \right )^\perp)=0$ and $A_2(\left (\mathcal C_{\{t_1,t_2\}} \right )^\perp)=2$.
The first three Pless power moments in (\ref{eq:PPM}) give
\begin{align*}
\left\{
  \begin{array}{l}
    A_{i_1} + A_{i_2} +A_{i_3} = 3^{2m-2}-1,   \\
    i_1 A_{i_1} + i_2  A_{i_2} + i_3 A_{i_3}  = 2\cdot 3^{2m-2-1}n,  \\
   i_1^2 A_{i_1} + i_2^2  A_{i_2} + i_3^2 A_{i_3}  =  3^{2m-2-2}\left [ n(4n+2)+4 \right ].
  \end{array}
\right.
\end{align*}
Solving this system of equations,
we get the weight distribution in Table   \ref{table:3W-2}.
\end{proof}

\begin{example}
Let $m=5$, $E=\{0,1\}$, $0\le t_1< t_2 \le \frac{3^m-1}{2}-1$ and $\mathcal C=\mathcal {C}(E)$. Then $\mathcal C$ has the weight distribution in Table \ref{table:3W}.
Furthermore, the shortened code $\mathcal C_{\{t_1, t_2\}}$ has parameters $[119,8,72]$ and weight enumerator $1+ 1764z^{72}+ 4238z^{81}+558z^{90}$.
This code has the same parameters as the best ternary linear code known in the database maintained  by Markus Grassl.

Magma experiments showed that all the shortened codes $\mathcal C_{\{t_1, t_2\}}$ have the
same weight distribution. However, for many pairs of $(t_1, t_2)$ and $(t'_1, t'_2)$, the codes
$\mathcal C_{\{t_1, t_2\}}$ and $\mathcal C_{\{t'_1, t'_2\}}$ are not equivalent. Therefore, the automorphism group of the
code $\mathcal{C}$ is in general not $2$-homogeneous and $2$-transitive.
\end{example}

To determine the weight distributions of some punctured codes from $\mathcal C(E)$, we need the next lemma.

\begin{lemma}\label{lem:W-C-T}
Let $\mathcal C$ be an $[n,k,d]$ code over $\mathrm{GF}(q)$ and let $d^{\perp}$ denote the minimum distance of $\mathcal  C^{\perp}$.
Let $t$ be a positive  integer and
$T$  a subset of the  coordinate positions of $\mathcal C$,
where $t<d^{\perp}$ and $\# T\le t$.
Suppose that  $A_i(\mathcal C_{T})$
 is  independent of
the elements of $T$ and depends only on the size of $T$.
Define
 \begin{align*}
 W_i(\mathcal C,T)=\left \{\mathbf c=(c_0, \dots, c_{n-1}) \in \mathcal C: \mathrm{wt}(\mathbf c)=i, c_j\neq 0 \text{ for all } j \in T \right \}.
 \end{align*}
 Then, $\# W_i(\mathcal C,T)$ is independent of
the elements of $T$ and depends only on the size of $T$. Moreover,
\begin{align*}
\# W_i(\mathcal C,T)=A_i(\mathcal C)- \sum_{j=1}^{\# T} (-1)^{j-1} \binom{\# T}{j} A_i(\mathcal C_{\{0, 1, \dots, j-1\}}).
\end{align*}
\end{lemma}

\begin{proof}
By the inclusion-exclusion principle, one has
\begin{align*}
\# W_i(\mathcal C,T)=A_i(\mathcal C)- \sum_{j=1}^{\# T} (-1)^{j-1} \sum_{J\subseteq T, \# J=j} A_i(\mathcal C(J)).
\end{align*}
By assumption, $A_i(\mathcal C(J))=A_i(\mathcal C_{J})= A_i(\mathcal C_{\{0,1, \dots, \#J-1\}})$. The desired conclusions then follow.
\end{proof}

For $T=\{t\}$ and
$T=\{t_1,t_2\}$, we determine the weight distribution of the punctured code $\mathcal C^T$
from some subcodes of  $\mathcal C \left (0,1, \dots, \lfloor \frac{m}{2} \rfloor \right )$.
\begin{theorem}\label{thm:pc-1}
Let $t$ be an integer  and  $m\ge 3$  odd,
where $0 \le t \le \frac{3^m-1}{2}-1$. Let $\mathcal C$ be a subcode of $\mathcal C \left (0,1, \dots, \lfloor \frac{m}{2} \rfloor \right )$
with the weight distribution in Table  \ref{table:3W}.
Then, the punctured  code $\mathcal C^{\{t\}}$ is a ternary linear code of length $\frac{3^m-1}{2}-1$ and dimension $2m$, and has the weight distribution in Table \ref{table:6W-1}.

\begin{table}[htbp]
\centering
\caption{The weight distribution of the punctured code  $\mathcal C^{\{t\}}$}
\label{table:6W-1}
\begin{tabular}{|c|c|}
  \hline
  Weight & Frequency \\
  \hline
  $0$& $1$\\
\hline
$3^{m-1}-3^{\frac{m-1}{2}}$ & $ \frac{1}{2} \cdot \left ( 3^{m-1} +2\cdot 3^{\frac{m-1}{2}}-1  \right ) \left ( 3^{m-1} +3^{\frac{m-1}{2}}\right )     $\\
\hline
$3^{m-1}-3^{\frac{m-1}{2}}-1$ & $ 3^{m-1}(3^{m-1}-1)    $\\
\hline
$3^{m-1}$&
 $\left ( 2\cdot 3^{m-1} + 1\right ) \left ( 3^{m-1} -1\right )  $\\
 \hline
 $3^{m-1}-1$&
 $2\cdot 3^{m-1} \left ( 2\cdot 3^{m-1} + 1\right )  $\\
 \hline
 $3^{m-1}+3^{\frac{m-1}{2}}$ & $ \frac{1}{2} \cdot \left ( 3^{m-1} -2\cdot 3^{\frac{m-1}{2}}-1  \right ) \left ( 3^{m-1} -3^{\frac{m-1}{2}}\right )      $\\
\hline
$3^{m-1}+3^{\frac{m-1}{2}}-1$ & $ 3^{m-1}(3^{m-1}-1)    $\\
\hline
\end{tabular}
\end{table}
\end{theorem}

\begin{proof}
It follows from Lemma \ref{eq:dual-abc} that $d(\mathcal C^\perp)=4$.
According to \cite{HP10}, $\mathcal C^{\{t\}}$ has length $n=\frac{3^m-1}{2}-1$ and dimension $k=2m$.
For $i\in \{3^{m-1}-3^{\frac{m-1}{2}}, 3^{m-1},  3^{m-1}+3^{\frac{m-1}{2}} \}$, by the definition of $\mathcal C^{\{t\}}$, we have
\begin{align*}
A_i(\mathcal C^{\{t\}})=A_i(\mathcal C_{\{t\}})
\end{align*}
and
\begin{align*}
A_{i-1}(\mathcal C^{\{t\}})=A_i(\mathcal C)-A_i(\mathcal C_{\{t\}}).
\end{align*}
The desired conclusions then follow from Theorem \ref{thm:sc1}.
\end{proof}

\begin{example}
Let $m=5$, $E=\{0,1\}$, $0\le t \le \frac{3^m-1}{2}-1$ and $\mathcal C=\mathcal {C}(E)$. Then
 the punctured  code $\mathcal C^{\{t\}}$ has parameters $[120,10,71]$ and weight enumerator $1+
6480z^{71}+ 4410 z^{72}+
 26406 z^{80}+
13040 z^{81}+6480 z^{89}+2232 z^{90}$.
This code has the same parameters as the best ternary linear code known in the database maintained  by Markus Grassl.

All the punctured codes $\mathcal C^{\{t\}}$ are equivalent,
as the automorphism group of the code $\mathcal{C}$ is transitive.
\end{example}

\begin{theorem}\label{thm:pc-2}
Let $t_1$ and $t_2$ be two integers  and  $m \ge 3$ be odd,
where $0 \le t_1 < t_2 \le \frac{3^m-1}{2}-1$. Let $\mathcal C$ be a  subcode of $\mathcal C \left (0,1, \dots, \lfloor \frac{m}{2} \rfloor \right )$
with the weight distribution in Table  \ref{table:3W}.
Then, the punctured code $\mathcal C^{\{t_1, t_2\}}$ is a ternary linear code of length $\frac{3^m-1}{2}-2$ and dimension $2m$, and has the weight distribution in Table \ref{table:9W-2}.

\begin{table}[htbp]
\centering
\caption{The weight distribution of the punctured code $\mathcal C^{\{t_1, t_2\}}$}
\label{table:9W-2}
\begin{tabular}{|c|c|}
  \hline
  Weight & Frequency \\
  \hline
  $0$& $1$\\
\hline
$3^{m-1}-3^{\frac{m-1}{2}}$ & $ \frac{1}{6} \cdot \left ( 3^{m-1} +2\cdot 3^{\frac{m-1}{2}}-1  \right ) \left ( 3^{m-1} +3^{\frac{m+1}{2}}\right )     $\\
\hline
$3^{m-1}-3^{\frac{m-1}{2}}-1$ & $ 2 \cdot  3^{m-2}  \left ( 3^{m-1} +2\cdot 3^{\frac{m-1}{2}}-1  \right )      $\\
\hline
$3^{m-1}-3^{\frac{m-1}{2}}-2$ & $ 2 \cdot  3^{m-2} \left ( 3^{m-1} - 3^{\frac{m-1}{2}}-1  \right )    $\\
\hline
$3^{m-1}$&
 $\left ( 2\cdot 3^{m-1} + 1\right ) \left ( 3^{m-2} -1\right )  $\\
 \hline
 $3^{m-1}-1$&
 $4 \cdot 3^{m-2} \left ( 2\cdot 3^{m-1} + 1\right )  $\\
 \hline
 $3^{m-1}-2$&
 $4 \cdot 3^{m-2} \left ( 2\cdot 3^{m-1} + 1\right )  $\\
 \hline
 $3^{m-1}+3^{\frac{m-1}{2}}$ & $ \frac{1}{6} \cdot \left ( 3^{m-1} -2\cdot 3^{\frac{m-1}{2}}-1  \right ) \left ( 3^{m-1} -3^{\frac{m+1}{2}}\right )      $\\
\hline
$3^{m-1}+3^{\frac{m-1}{2}}-1$ & $  2 \cdot  3^{m-2}  \left ( 3^{m-1} - 2\cdot 3^{\frac{m-1}{2}}-1  \right )       $\\
\hline
$3^{m-1}+3^{\frac{m-1}{2}}-2$ & $  2 \cdot  3^{m-2} \left ( 3^{m-1} + 3^{\frac{m-1}{2}}-1  \right )       $\\
\hline
\end{tabular}
\end{table}
\end{theorem}

\begin{proof}
By Lemma \ref{eq:dual-abc}, $d(\mathcal C^\perp)=4$. According to \cite{HP10},
$\mathcal C^{\{t_1, t_2\}}$ has length $n=\frac{3^m-1}{2}-2$ and dimension $k=2m$.
For $i \in \{3^{m-1}-3^{\frac{m-1}{2}}, 3^{m-1},  3^{m-1}+3^{\frac{m-1}{2}} \}$, using the definition of $\mathcal C^{\{t_1, t_2\}}$, we deduce that
\begin{align*}
\left\{
  \begin{array}{l}
    A_i(\mathcal C^{\{t_1,t_2\}})=A_i(\mathcal C_{\{t_1,t_2\}}),  \\
   A_{i-1}(\mathcal C^{\{t_1,t_2\}})=A_i(\mathcal C)-A_i(\mathcal C_{\{t_1,t_2\}})-\# W_i(\mathcal C, \{t_1,t_2\}) , \\
   A_{i-2}(\mathcal C^{\{t_1,t_2\}})=\# W_i(\mathcal C, \{t_1,t_2\}).
  \end{array}
\right.
\end{align*}
It then follows from Lemma \ref{lem:W-C-T} that
\begin{align*}
\# W_i(\mathcal C, \{t_1,t_2\})=A_i(\mathcal C)- 2 A_i(\mathcal C_{\{0\}}) +A_i(\mathcal C_{\{0,1\}}).
\end{align*}
The desired conclusions then follow from Theorems \ref{thm:sc1} and \ref{thm:sc2}.
\end{proof}

\begin{example}
Let $m=5$, $E=\{0,1\}$, $0\le t_1< t_2 \le \frac{3^m-1}{2}-1$ and $\mathcal C=\mathcal {C}(E)$. Then
the punctured  code $\mathcal C^{\{t_1, t_2\}}$ has parameters $[119,10,70]$ and weight enumerator $1+ 3834z^{70}+
5292z^{71}+1764z^{72}+
17604z^{79}+
 17604z^{80}+
4238 z^{81}+
4806 z^{88}+3348 z^{89}+558z^{90}$.
This code has the same parameters as the best ternary linear code known in the database maintained  by Markus Grassl.

Our Magma experiments showed that for many different pairs of $\{t_1, t_2\}$ and
$\{t'_1, t'_2\}$, the punctured codes $\mathcal C^{\{t_1, t_2\}}$ and
$\mathcal C^{\{t'_1, t'_2\}}$
are not equivalent. Hence, the automorphism group of the
code $\mathcal{C}$ is in general not $2$-homogeneous and $2$-transitive.
\end{example}

\section{Steiner systems and $2$-designs from  $\mathcal{C}(E)$}

In this section, we confirm Conjectures 1 and 2, and Conjecture 3 for $k \in \{4,5,6,7\}$.
In addition, we construct more $2$-designs from subcodes of $\mathcal C \left (0,1, \dots, \lfloor \frac{m}{2} \rfloor \right )$.

Let $\mathcal C$ be an $[n,k,d]$ linear code. Define
\begin{align*}
W_i(\mathcal C)=\{\mathbf c\in \mathcal C: \mathrm{wt}(\mathbf c)=i\}, ~~~~0\le i \le n.
\end{align*}

\begin{theorem}\label{thm:Steiner-A4}
Let $m\ge 3$ be a positive integer and $\mathcal C$ a subcode of $\mathcal C \left (0,1, \dots, \lfloor \frac{m}{2} \rfloor \right )$
 such that $A_4(\mathcal C^{\perp})=\frac{(3^m-1)(3^{m-1}-1)}{8}.$
Then $(\mathcal P(\mathcal C^{\perp}), \mathcal B_{4}(\mathcal C^{\perp}) )$
is the Steiner system $\mathbb{PD}(1,m-1,3)$ with parameters $S(2,4,\frac{3^m-1}{2})$.
\end{theorem}
\begin{proof}
Note that  $\mathcal C \left (0,1, \dots, \lfloor \frac{m}{2} \rfloor \right )^{\perp} \subseteq \mathcal C^{\perp}$.
Thus, $ W_{4}\left (\mathcal C \left (0,1, \dots, \lfloor \frac{m}{2} \rfloor \right )^{\perp} \right ) \subseteq  W_4 \left ( \mathcal C^{\perp} \right )$.
From Part  (iv) of Theorem  \ref{thm:C(0,1,...,m/2)} and $A_4(\mathcal C^{\perp})=\frac{(3^m-1)(3^{m-1}-1)}{8}$, we have $W_{4}\left (\mathcal C \left (0,1, \dots, \lfloor \frac{m}{2} \rfloor \right )^{\perp} \right ) =  W_4 \left ( \mathcal C^{\perp} \right )$ and $ \mathcal B_4 \left ( \mathcal C^{\perp} \right )=\mathcal B_{4}\left (\mathcal C \left (0,1, \dots, \lfloor \frac{m}{2} \rfloor \right )^{\perp} \right ) $. The conclusions of this theorem finally follow from
Part (iii) of Theorem \ref{thm:C(0,1,...,m/2)}.
\end{proof}

\begin{corollary}\label{cor:steiner}
Let   $m\ge 3$ be an odd integer and $\mathcal C$ a subcode of $\mathcal C \left (0,1, \dots, \lfloor \frac{m}{2} \rfloor \right )$
with the weight distribution in Table  \ref{table:3W}. If $\mathcal C'$ is a linear  code such that
$\mathcal C \subseteq \mathcal C' \subseteq \mathcal C \left (0,1, \dots, \lfloor \frac{m}{2} \rfloor \right )$, then $(\mathcal P(\mathcal C'^{\perp}), \mathcal B_{4}(\mathcal C'^{\perp}) )$
is the Steiner system $\mathbb{PD}(1,m-1,3)$ with parameters $S(2,4,\frac{3^m-1}{2})$.
\end{corollary}

\begin{proof}
Note that $W_{4}\left (\mathcal C \left (0,1, \dots, \lfloor \frac{m}{2} \rfloor \right )^{\perp} \right )\subseteq  W_4 \left ( \mathcal C'^{\perp} \right ) \subseteq  W_4 \left ( \mathcal C^{\perp} \right )$.
From Part (iv) of Theorem \ref{thm:C(0,1,...,m/2)} and Lemma \ref{eq:dual-abc}, we then deduce that  $A_4(\mathcal C'^{\perp})=\frac{(3^m-1)(3^{m-1}-1)}{8}$. From Theorem \ref{thm:Steiner-A4}, this corollary follows.
\end{proof}


\begin{remark}
Corollary \ref{cor:steiner} confirmed Conjecture \ref{conj:Steiner}, and proved that the
Steiner system $\mathbb{PD}(1,m-1,3)$ is supported by many ternary linear codes. It is an
interesting problem to find a linear code that supports a given design. This is in general
a hard problem. In addition, Corollary \ref{cor:steiner} says that the duals of many subcodes of $\mathcal C \left (0,1, \dots, \lfloor \frac{m}{2} \rfloor \right )$ do not support a new Steiner system, but the geometric Steiner system $\mathbb{PD}(1,m-1,3)$.
\end{remark}

Let $\mathbb{D}$ be a $t$-$(v, k, \lambda)$ design. For a majority decoding of the code
$\mathbf{C}_q(\mathbb{D})^\perp$, Tonchev introduced the \emph{dimension} of $\mathbb{D}$ over
GF($q$), which is defined to be minimum dimension of all linear codes of length $v$
over GF($q$) that contain the blocks of $\mathbb{D}$ as the supports of codewords of
weight $k$ \cite{Ton99}. When $q=2$, the dimension of $\mathbb{D}$ over GF($q$) is the
same as the rank of $\mathbb{D}$ over GF($q$). When $q>2$, the rank of $\mathbb{D}$
over GF($q$) is an upper bound of the dimension of $\mathbb{D}$ over GF($q$).
Since $\mathcal C \left (0,1, \dots, \lfloor \frac{m}{2} \rfloor \right )$ has dimension
$m(m+1)/2$, its dual has dimension $(3^m-1-m(m+1))/2$. As a result, the dimension of
$\mathbb{PD}(1,m-1,3)$ over GF($3$) is upper bounded by $(3^m-1-m(m+1))/2$. Then the
following open problem arises.

\begin{open}
Is the dimension of the Steiner system $\mathbb{PD}(1,m-1,3)$ over GF($3$) equal to
$(3^m-1-m(m+1))/2$.
\end{open}


We will need the following  lemma from \cite[p. 121]{Dingbk18}  to determine the parameters of some $t$-designs.

\begin{lemma}\label{lem:(q-1)d}
Let $\mathcal C$ be a linear code over $\mathrm{GF}(3)$ with minimum weight $d$.
Let $\mathbf c$ and  $\mathbf c'$ be two codewords of weight $i$  and  $\mathrm{Supp}(\mathbf c)=\mathrm{Supp}(\mathbf c')$,
where $d\le i \le 2d-1$.
Then $\mathbf c'= \mathbf c$ or $\mathbf c'=- \mathbf c$.
\end{lemma}

\begin{theorem}\label{thm:3-2-designs}
Let   $m\ge 5$ be an odd integer and $\mathcal C$ a subcode of $\mathcal C \left (0,1, \dots, \lfloor \frac{m}{2} \rfloor \right )$
with the weight distribution in Table  \ref{table:3W}. Let $k\in \left \{3^{m-1}-3^{\frac{m-1}{2}}, 3^{m-1},3 ^{m-1}+3^{\frac{m-1}{2}} \right \}$. Then $\left (\mathcal P(\mathcal C), \mathcal B_{k}(\mathcal C) \right )$
is a $2$-$(\frac{3^m-1}{2}, k, \lambda)$ design, where
\begin{align*}
\lambda= \frac{A_k(\mathcal C)+A_k \left ( \mathcal C_{\{0,1\}} \right )}{2}- A_k\left (\mathcal C_{\{0\}} \right ).
\end{align*}
\end{theorem}

\begin{proof}
Let $k\in \left \{3^{m-1}-3^{\frac{m-1}{2}}, 3^{m-1},3 ^{m-1}+3^{\frac{m-1}{2}} \right \}$ and $0\le i <j \le \frac{3^m-1}{2}-1$.
Define
\begin{align*}
\mathcal B_k \left ( \mathcal C, \{i,j\} \right )= \left \{ \mathrm{Supp}(\mathbf c):  \mathbf c \in W_k\left (\mathcal C, \{i,j\} \right )\right  \},
\end{align*}
where $W_k\left (\mathcal C, \{i,j\} \right )$ was defined in Lemma \ref{lem:W-C-T}.
Since $m\ge 5$, we have $k\le 2\cdot \left (3^{m-1}-3^{\frac{m-1}{2}}\right )-1$. From Lemma \ref{lem:(q-1)d}, we get
\begin{align*}
\# \mathcal B_k \left ( \mathcal C, \{i,j\} \right )= \frac{1}{2} \# W_k\left (\mathcal C, \{i,j\} \right ).
\end{align*}
Using Theorems \ref{thm:sc1} , \ref{thm:sc2} and  Lemma \ref{lem:W-C-T}, we obtain
\begin{align*}
\# \mathcal B_k \left ( \mathcal C, \{i,j\} \right )= \frac{A_k(\mathcal C)+A_k \left ( \mathcal C_{\{0,1\}} \right )}{2}- A_k\left (\mathcal C_{\{0\}} \right ).
\end{align*}
Therefore, $\# \mathcal B_k \left ( \mathcal C, \{i,j\} \right )$ is independent of $i$ and $j$. Consequently, the codewords of weight $k$ hold a $2$-design. This completes the proof.
\end{proof}

\begin{corollary}
Let   $m\ge 5$ be an odd integer and $\mathcal C$ a subcode of $\mathcal C \left (0,1, \dots, \lfloor \frac{m}{2} \rfloor \right )$
with the weight distribution in Table  \ref{table:3W}. Then $\mathcal C$ holds three $2$-$(\frac{3^m-1}{2}, k, \lambda)$ designs with the following
pairs $(k,\lambda)$:
\begin{itemize}
\item $\left (3^{m-1}-3^{\frac{m-1}{2}},  3^{m-2} \left ( 3^{m-1}-3^{\frac{m-1}{2}} -1 \right )\right )$.
\item $\left (3^{m-1}, 2 \cdot 3^{m-2} \left ( 2\cdot 3^{m-1} +1 \right )\right )$.
\item $\left (3^{m-1}+3^{\frac{m-1}{2}},  3^{m-2} \left ( 3^{m-1}+3^{\frac{m-1}{2}} -1 \right )\right )$.
\end{itemize}

\end{corollary}
\begin{proof}
From Theorems \ref{thm:sc1} , \ref{thm:sc2}, and \ref{thm:3-2-designs}, this corollary   follows. Alternatively, the conclusions of this corollary follow from Theorem \ref{thm:3-2-designs}
and Lemma \ref{lem:(q-1)d}.
\end{proof}

\begin{remark}
Theorem \ref{thm:3-2-designs} confirmed Conjecture \ref{conj:3 design}, and extends
Conjecture \ref{conj:3 design} if more subcodes of $\mathcal C \left (0,1, \dots, \lfloor \frac{m}{2} \rfloor \right )$
with the weight distribution in Table  \ref{table:3W} exist.
\end{remark}

\begin{lemma}\label{lem:Ak=A-lam}
Let   $m\ge 3$ be an odd integer and $\mathcal C$ a subcode of $\mathcal C \left (0,1, \dots, \lfloor \frac{m}{2} \rfloor \right )$
with the weight distribution in Table  \ref{table:3W}. Let $A=(3^m-1)(3^{m}-3)$ and $\lambda_k^{\perp}=\frac{2k(k-1)A_k\left ( \mathcal C^{\perp} \right )}{(3^m-1)(3^m-3)}$ with $0\le k \le \frac{3^m-1}{2}-1$.
Then,  $\left (A_k\left ( \mathcal C^{\perp} \right ),  \lambda_k^{\perp} \right )$  is given in Table \ref{table:AkDC-lam}, where $k\in \{4,5,6,7\}$.
\begin{table}[htbp]
\centering
\caption{$A_k\left ( \mathcal C^{\perp} \right )$ and $\lambda_k^{\perp} $ for $4 \le k \le 7$}
\label{table:AkDC-lam}
\begin{tabular}{|c|c|c|}
  \hline
  k &  $\lambda_k^{\perp} $ &$A_k\left ( \mathcal C^{\perp} \right )$  \\
  \hline
 $4$ &  $1$ & $\frac{1}{8}A$\\
\hline
$5$ & $ 3^{m-1}-9    $ & $\frac{1}{40}A \lambda_5^{\perp}$ \\
\hline
$6$&
 $ \frac{3}{4}\left (3^{2m-2} -38\cdot 3^{m-2}+53 \right ) $& $\frac{1}{60}A \lambda_6^{\perp}$ \\
 \hline
 $7$& $\frac{1}{20}\left ( 3^{3m-2} -5\cdot 3^{2m}+1006 \cdot 3^{m-2}-1000 \right )      $ & $\frac{1}{84}A \lambda_7^{\perp}$\\
\hline
\end{tabular}
\end{table}

\end{lemma}
\begin{proof}
The desired conclusions follow from Lemma \ref{eq:dual-abc}.
\end{proof}

\begin{theorem}\label{thm:dual-design4567}
Let   $m \ge 5$ be an odd integer and $\mathcal C$ a subcode of $\mathcal C \left (0,1, \dots, \lfloor \frac{m}{2} \rfloor \right )$
with the weight distribution in Table  \ref{table:3W}. Let $k\in \{4,5,6,7\}$. Then $\left (\mathcal P(\mathcal C^{\perp}), \mathcal B_{k}(\mathcal C^{\perp}) \right )$
is a $2$-$(\frac{3^m-1}{2}, k, \lambda^{\perp}_k)$ design, where $\lambda^{\perp}_k$ is given in Table \ref{table:AkDC-lam}.
\end{theorem}

\begin{proof}
Let $4\le k \le 7$ and $0\le i <j \le \frac{3^m-1}{2}-1$.
Define
\begin{align*}
\mathcal B_k \left ( \mathcal C^{\perp}, \{i,j\} \right )= \left \{ \mathrm{Supp}(\mathbf c):  \mathbf c \in W_k\left (\mathcal C^{\perp}, \{i,j\} \right )\right  \},
\end{align*}
where $W_k\left (\mathcal C^{\perp}, \{i,j\} \right )$ was defined in Lemma \ref{lem:W-C-T}.
It follows from Lemma \ref{eq:dual-abc} that $d(\mathcal C^{\perp})=4$ and $k\le 2d(\mathcal C^{\perp})-1$. From Lemma \ref{lem:(q-1)d}, we get
\begin{align*}
\# \mathcal B_k \left ( \mathcal C^{\perp}, \{i,j\} \right )=& \frac{1}{2} \# W_k\left (\mathcal C^{\perp}, \{i,j\} \right ).
\end{align*}
Using the inclusion-exclusion principle, one has
\begin{align}\label{eq:dual-IEP}
\# W_k\left (\mathcal C^{\perp}, \{i,j\} \right )=A_k \left (\mathcal C^{\perp} \right )-
A_k \left ( \left ( \mathcal C^{\perp} \right )_{\{i\}} \right )-A_k \left ( \left ( \mathcal C^{\perp} \right )_{\{j\}} \right )+
A_k \left ( \left ( \mathcal C^{\perp} \right )_{\{i,j\}} \right ).
\end{align}
From Theorems \ref{thm:pc-1} and \ref{thm:pc-2}, for any $0\le i \le \frac{3^m-1}{2}-1$, one has
\begin{align*}
A_i \left ( \mathcal C^{ \{j_0\} } \right )= A_i\left ( \mathcal C^{\{0\}} \right ), \  A_i \left ( \mathcal C^{\{j_0,j_1\}} \right )= A_i\left ( \mathcal C^{\{0,1\}} \right ).
\end{align*}
Using the MacWilliams Identity, one gets
\begin{align*}
A_i \left ( \left ( \mathcal C^{ \{j_0\} } \right )^{\perp}\right )= A_i\left ( \left ( \mathcal C^{\{0\}} \right )^{\perp} \right ), \
 A_i \left ( \left ( \mathcal C^{\{j_0,j_1\}} \right )^{\perp} \right )= A_i\left ( \left ( \mathcal C^{\{0,1\}}\right )^{\perp} \right ),
\end{align*}
where $0\le i \le \frac{3^m-1}{2}-1$.
From (\ref{eq:sh-puct}),  $\left ( \mathcal C^{\perp} \right )_{\{T\}}=\left (\mathcal C^{T} \right )^{\perp}$ for any $T\subseteq  \{0,1, \dots, \frac{3^m-1}{2}-1\}$.
Thus,
\begin{align*}
A_i \left ( \left ( \mathcal C ^{\perp}\right )  _{ \{j_0\} } \right )= A_i\left ( \left ( \mathcal C ^{\perp}  \right ) _{\{0\}} \right ), \
 A_i \left ( \left ( \mathcal C ^{\perp} \right )  _{\{j_0,j_1\}} \right )= A_i\left ( \left ( \mathcal C  ^{\perp} \right ) _{\{0,1\}} \right ),
\end{align*}
where $0\le i \le \frac{3^m-1}{2}-1$.
From (\ref{eq:dual-IEP}), one obtains
\begin{align*}
\# W_k\left (\mathcal C^{\perp}, \{i,j\} \right )=A_k \left (\mathcal C^{\perp} \right )-
2A_k \left ( \left ( \mathcal C^{\perp} \right )_{\{0\}} \right )+
A_k \left ( \left ( \mathcal C^{\perp} \right )_{\{0,1\}} \right ).
\end{align*}
Then,
\begin{align*}
\# \mathcal B_k \left ( \mathcal C^{\perp}, \{j_0,j_1\} \right )=& \frac{A_k \left (\mathcal C^{\perp} \right )+
A_k \left ( \left ( \mathcal C^{\perp} \right )_{\{0,1\}} \right )}{2}- A_k \left ( \left ( \mathcal C^{\perp} \right )_{\{0\}} \right ).
\end{align*}
Therefore, $\# \mathcal B_k \left ( \mathcal C^{\perp}, \{j_0,j_1\} \right )$ is independent of $j_0$ and $j_1$. Hence, the codewords of weight $k$ hold a
$2$-$(\frac{2^m-1}{3},k, \lambda^{\perp}_k)$ design with $A_k \left (\mathcal C^{\perp} \right )$ blocks.
Thus,
\begin{align*}
\lambda_k^{\perp}= \frac{2k(k-1)A_k\left ( \mathcal C^{\perp} \right )}{(3^m-1)(3^m-3)}.
\end{align*}
The desired conclusions finally follow from Lemma \ref{lem:Ak=A-lam}.
\end{proof}

\begin{remark}
Theorem \ref{thm:dual-design4567} confirmed Conjecture \ref{conj:many designs} for $k=4,5,6$ and $7$.
\end{remark}

\section{Summary}

In this paper, from a ternary linear subcode  of $\mathcal C(0,1,\dots, \frac{m-1}{2})$
with the weight distribution in Table \ref{table:3W}, some infinite families of $2$-designs
with various block sizes were confirmed and their parameters were settled.
Notice that the automorphism group of a ternary code with the weight distribution
in Table \ref{table:3W} is in general neither $2$-transitive nor $2$-homogeneous. In
addition, such codes do not satisfy the conditions in the Assmus-Mattson Theorem.
Hence, we had to use a direct approach to proving that the codes hold $2$-designs.
This makes the
designs presented in this paper very interesting, as such non-symmetric designs
are rare.

Another contribution of this paper is the construction of three-weight, six-weight and
nine-weight ternary linear codes by shortening and puncturing the ternary codes with the weight distribution in Table \ref{table:3W}. The parameters of these codes appear to be new.
These codes include  some optimal codes with the same parameters as the best ternary linear codes  known in the database maintained  by Markus Grassl at http://www.codetables.de/.
These ternary codes  can be employed to obtain secret sharing schemes with interesting
access structures using the framework  in \cite{ADHK98,YD06}.

It is still open if a ternary linear code with the weight distribution in Table $\ref{table:3W}$ is always equivalent to a subcode of $\mathcal C(0,1,\dots, \frac{m-1}{2})$.

\textbf{Acknowledgements}
C. Tang was supported by National Natural Science Foundation of China (Grant No.
11871058) and China West Normal University (14E013, CXTD2014-4 and the Meritocracy Research
Funds).
C. Ding was supported by The Hong Kong Research Grants Council, Project No. 16300418.
M. Xiong was supported by The Hong Kong Research Grants Council, Project No. NHKUST619/17.



\end{document}